\documentclass[sigconf]{aamas}

\usepackage{balance}

\usepackage{amsmath,amssymb,amsfonts,amsthm}
\usepackage{algorithm, algorithmic}
\usepackage{graphicx}
\usepackage{acronym}
\PassOptionsToPackage{hyphens}{url}
\usepackage{hyperref}
\makeatletter
\@ifundefined{c@biburlnumpenalty}{}{\setcounter{biburlnumpenalty}{9000}}
\@ifundefined{c@biburlucpenalty}{}{\setcounter{biburlucpenalty}{9000}}
\@ifundefined{c@biburllcpenalty}{}{\setcounter{biburllcpenalty}{9000}}
\makeatother
\usepackage{textcomp}
\usepackage{subcaption}
\usepackage{booktabs}
\usepackage[switch]{lineno}
\usepackage{acronym}
\usepackage[english]{babel}
\usepackage{paralist}
\definecolor{darkgreen}{rgb}{0.0, 0.5, 0.0}
\usepackage{todonotes}
\usepackage{mathtools}
\usepackage{bm}
\usepackage{xcolor}
\usepackage{xurl}

\doi{STZK9664}

\makeatletter
\gdef\@copyrightpermission{
  \begin{minipage}{0.2\columnwidth}
   \href{https://creativecommons.org/licenses/by/4.0/}{\includegraphics[width=0.90\textwidth]{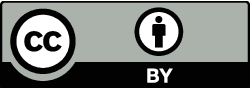}}
  \end{minipage}\hfill
  \begin{minipage}{0.8\columnwidth}
   \href{https://creativecommons.org/licenses/by/4.0/}{This work is licensed under a Creative Commons Attribution International 4.0 License.}
  \end{minipage}
  \vspace{5pt}
}
\makeatother

\setcopyright{ifaamas}
\acmConference[AAMAS '26]{Proc.\@ of the 25th International Conference
on Autonomous Agents and Multiagent Systems (AAMAS 2026)}{May 25 -- 29, 2026}
{Paphos, Cyprus}{C.~Amato, L.~Dennis, V.~Mascardi, J.~Thangarajah (eds.)}
\copyrightyear{2026}
\acmYear{2026}
\acmDOI{}
\acmPrice{}
\acmISBN{}

\title[AAMAS-2026 Formatting Instructions]{IMAS\texorpdfstring{$^2$}{2}: Joint Agent Selection and Information-Theoretic Coordinated Perception In Dec-POMDPs}

\author{Chongyang Shi}
\affiliation{
  \institution{University of Florida}
  \city{Gainesville}
  \country{United States}}
\email{c.shi@ufl.edu}

\author{Wesley A. Suttle}
\affiliation{
  \institution{Army Research Laboratory}
  \city{Adelphi}
  \country{United States}}
\email{wesley.a.suttle.ctr@army.mil}

\author{Michael Dorothy}
\affiliation{
  \institution{Army Research Laboratory}
  \city{Adelphi}
  \country{United States}}
\email{michael.r.dorothy.civ@army.mil}

\author{Jie Fu}
\affiliation{
  \institution{University of Florida}
  \city{Gainesville}
  \country{United States}}
\email{fujie@ufl.edu}

\begin{abstract}
We study the problem of jointly selecting sensing agents and synthesizing decentralized active perception policies for the chosen subset of agents within a Decentralized Partially Observable Markov Decision Process (Dec-POMDP) framework. 
Our approach employs a two-layer optimization structure. In the inner layer, we introduce information-theoretic metrics, defined by the mutual information between the unknown trajectories or some hidden property in the environment and the collective partial observations in the multi-agent system, as a unified objective for active perception problems. We employ various optimization methods to obtain optimal sensor policies that maximize mutual information for distinct active perception tasks. In the outer layer, we prove that under certain conditions,  the information-theoretic objectives are monotone and submodular with respect to the subset of observations collected from multiple agents. We then exploit this property to design an IMAS$^2$ (Information-theoretic Multi-Agent Selection and Sensing) algorithm for joint sensing agent selection and sensing policy synthesis. However, since the policy search space is infinite, we adapt the classical Nemhauser-Wolsey argument to prove that the proposed IMAS$^2$ algorithm can provide a tight $(1 - 1/e)$-guarantee on the performance.
Finally, we demonstrate the effectiveness of our approach in a multi-agent cooperative perception in a grid-world environment.
\end{abstract}


\newcommand{\BibTeX}{\rm B\kern-.05em{\sc i\kern-.025em b}\kern-.08em\TeX}

  \newcommand{\nat}{\mathbf{N}}

\makeatletter
\@ifundefined{cf}{
    
}{
    
}
\@ifundefined{eg}{
    \newcommand{\eg}{e.g.}
}{
    \renewcommand{\eg}{e.g.}
}
\@ifundefined{ie}{
    \newcommand{\ie}{i.e.~}
}{
    \renewcommand{\ie}{i.e.~}
}
\@ifundefined{etc}{
    
}{
    
}
\@ifundefined{etal}{
    \newcommand{\etal}{et~al.}
}{
    \renewcommand{\etal}{et~al.}
}
\makeatother

\newcommand{\dist}{\mathcal{D}}

\newcommand{\reals}{\mathbb{R}}

\newtheorem{theorem}{Theorem} 
\newtheorem{lemma}{Lemma} 

\newtheorem{proposition}{Proposition}
\newtheorem{corollary}{Corollary} 

\newtheorem{problem}{Problem}

\newtheorem{remark}{Remark}

\newtheorem{assumption}{Assumption}

\acrodef{mdp}[MDP]{Markov decision process} 
\acrodef{lmdp}[LMDP]{labeled Markov decision process} 
\acrodef{dec_pomdp}[Dec-POMDP]{decentralized partially observable Markov decision process} 

\acrodef{asw}[ASW]{Almost-Sure Winning}
\acrodef{ltlf}[LTL$_f$]{Linear Temporal Logic over Finite Traces}
\acrodef{ltl}[LTL]{linear temporal logic}
\acrodef{scltl}[co-safe LTL]{syntactically co-safe Linear Temporal Logic}
\acrodef{dfa}[DFA]{deterministic finite automaton}
\acrodef{des}[DES]{discrete event system}

 \DeclareMathOperator*{\optmax}{\mathrm{maximize}}

 \newcommand{\argmax}{\mathop{\mathrm{argmax}}}

\newcommand{\expect}{\mathbb{E}}

\newcommand{\calM}{\mathcal{M}}

\newcommand{\calS}{\mathcal{S}}
\newcommand{\calA}{\mathcal{A}}

\newcommand{\calO}{\mathcal{O}}

\acrodef{hmm}[HMM]{hidden Markov model}
\acrodef{fsc}[FSC]{finite state controller}
\acrodef{pomdp}[POMDP]{partially observable Markov decision process} 

\newcommand{\calK}{{\mathcal{K}}}

\begin{document}


\pagestyle{fancy}
\fancyhead{}


\maketitle 


\section{Introduction}
Autonomous multi-agent systems are increasingly deployed in environments where agents must actively gather information under uncertainty. Applications include teams of robots for surveillance and search-and-rescue \cite{Rybski2000surveillance, Kantor2006, Rosenfeld2015Intelligent}, sensor networks for target tracking \cite{Tian2006Target, Jing2020Pose},  and cooperative perception in autonomous driving and aerial monitoring \cite{Xu2020hitmac, wang2023corecooperativereconstructionmultiagent}. 
In such settings, an active perception agent must decide not only how to act, but also what to sense, given the past observation. In multi-agent cooperative active perception, it must additionally account for its knowledge of other agents’ possible observations to avoid redundancy and achieve better coordination. Meanwhile, in addition to perception, a multi-agent system may be required to perform tasks leveraging the information from the perception agents, for example, in a multi-UAV target tracking problem, a subset of UAVs may be tasked to ensure accurate tracking of the moving target, while other UAVs may be tasked to interdict the target given the target trajectory.

Motivated by these\ applications, this paper studies the following research problem: ``Given a multi-agent system with heterogeneous dynamics and perception capabilities, how to select a subset of agents and design their decentralized  perception strategies for a perception objective?'' Specifically, we model a multi-agent system and its interaction with the dynamic, stochastic environment using a \ac{dec_pomdp} framework. We aim to select a subset of the agents for perception tasks,  whose objectives are defined by maximizing the mutual information between some unknown quantity (trajectory, states, or critical events) and the collective observations of the perception team.

\noindent\textbf{Related work.} 
This problem is closely related to the sensor placement problem, which aims to find an optimal subset of agents to participate in a perception task under limited sensing resources. Unlike traditional sensor coverage problems that emphasize spatial reach or visibility, the goal is to maximize the informativeness of the selected agents’ observations about the environment or latent variables of interest. Selecting too many agents wastes resources and increases redundancy, while selecting too few may lead to poor estimation or inference accuracy. To balance these trade-offs, many existing approaches formulate the objective using information-theoretic criteria, such as mutual information or entropy reduction \cite{Krause2007Near, Singh2009informative}. When the agents’ policies are finite, the environment is deterministic, and the inference objective is stationary—as in multi-robot path planning for predicting algae content in a lake—the resulting objectives often exhibit submodularity. This property enables the design of greedy selection algorithms with provable near-optimal performance guarantees and strong scalability for large multi-agent systems. However, in stochastic multi-agent systems, where the policy space is infinite (often parameterized by deep neural networks) and the unknown quantities to be inferred can be a stochastic process (such as tracking a moving agent), existing results do not directly apply.


The problem of decentralized active sensing and perception has been studied extensively in multi-agent systems. Early work, such as \cite{Kreucher2005Sensor}, formulated sensor management as a centralized active sensing problem, where a single decision maker plans sensing actions to maximize expected information gain using Bayesian filtering. While effective for small-scale systems, such centralized formulations are computationally expensive and scale poorly with the number of agents.  
Subsequent research has focused on decentralized and scalable planning methods.
Satsangi \etal~\cite{Satsangi_Whiteson_Oliehoek_2015, Satsangi2020Active} studied a class of partially observable Markov decision processes (POMDPs) whose value functions exhibit submodularity, and leveraged this property to design efficient algorithms for dynamic sensor selection with near-optimal guarantees. In contrast, Kumar \etal ~\cite{Kumar2017Decentralized} considered multi-agent planning under stochastic dynamics, where submodularity appears in the reward structure rather than the value function. Their framework provides a theoretical foundation for decentralized decision making with submodular rewards, enabling agents to coordinate without centralized control.
Best \etal ~\cite{Best2019MCTS} proposed Dec-MCTS, a decentralized Monte Carlo tree search method, for multi-robot active perception, which can handle stochasticity, but does not exploit the submodular structure of reward functions for efficiency or near-optimality guarantees. 
Lauri and Oliehoek~\cite{Lauri2020prediction} introduced a prediction-reward framework that quantifies the uncertainty in a centralized state estimate obtained after all agents complete sensing. To enable decentralized computation, they approximate this centralized reward using the expected accuracy of each agent’s prediction of the final belief. This reformulation expresses the global information objective as a standard \ac{dec_pomdp} reward function dependent only on local states, actions, and prediction actions, allowing the use of existing \ac{dec_pomdp} solvers while preserving the informativeness objective. However, this method also does not employ submodularity for more efficient computation.

Recent research also studied the multi-robot multi-target tracking problem in deterministic, nonlinear systems \cite{corahScalableDistributedPlanning2021}. They employ the receding horizon planning framework and consider at each planning epoch, the team of robots aims to optimize the mutual information between the target states and the observations generated by a set of feasible trajectories for each robot. Because they consider a finite set of trajectories for each robot and deterministic dynamics, such an objective function can be shown to be submodular and monotone. By leveraging the property of a submodular function, they develop an algorithm to efficiently compute the near-optimal trajectories for multiple robots.

\noindent\textbf{Our contributions.} 
Existing work on sensor selection or motion planning of active sensing agents considers finitely many sensors \cite{Singh2009informative} or discrete trajectories \cite{corahScalableDistributedPlanning2021}. In our setting, though the set of agents is finite, each agent is to compute an observation-based stochastic policy, and therefore, the policy space for agents is infinite. Therefore, existing GreedyMax algorithms for submodular \cite{nemhauser1978analysis} cannot be directly applied to compute jointly the subset of agents and their decentralized perception policies. Our problem can not be directly mapped to decentralized planning for multi-agent systems given submodular rewards because the information objectives cannot be decomposed into submodular reward given the joint state and joint actions (see Section~\ref{sec:submodularity}). Existing \ac{dec_pomdp} for active perception \cite{Best2019MCTS,Lauri2020prediction} does not solve the agent selection problem jointly with the policy synthesis.

First, we prove that in a \ac{dec_pomdp}, when the set of policies are fixed, then the mutual information between the joint state trajectory given a subset of agents' policies, and consequently their partial observations, is monotone and submodular, provided that the agents' observations are conditionally independent given the joint trajectory. Building on this result, we then investigate two other inference objectives: 
\begin{inparaenum}
\item Inferring the trajectory of an environment agent;
\item and inferring a secret property which is a function of the trajectory of an environment agent.
\end{inparaenum}
We show that in both cases, under certain assumption, the mutual information between the quantity to be inferred and a subset of agents' observations are submodular and monotone.

However, these results only enable us to use approximate algorithms for submodular maximization to solve the agent selection problem, assuming each selected agent follows a pre-defined decentralized policy. To achieve joint agent selection and policy synthesis,  we propose the IMAS$^2$ (Information-theoretic Multi-Agent Selection and Sensing) algorithm, a variant of the GreedyMax algorithm \cite{nemhauser1978analysis} to select one agent and determine its perception policy at a time based on the principle of maximizing the marginal gain. We then prove that, under additional constraints on subsequent maximal marginal gains,  this algorithm can provide a strong $(1-1/e)$-guarantee on the performance. The proposed algorithm achieves scalability and near-optimal performance, bridging rigorous submodular optimization theory with practical decentralized planning in multi-agent active perception.  
Finally, we formulate the agent selection and policy optimization problem as a two-layer optimization: an inner layer that computes an optimal perception policy, with respect to maximizing information gain of a chosen agent, in addition to observations of selected agents. And an outer layer efficiently selects $k$ agents to achieve the approximately optimal cooperative active perception.
The results are then experimentally validated.

\noindent\textbf{Assumptions and Scope.}
We focus on Dec-POMDPs with observation independent sensing models, formalized by Assumption~\ref{assumption:independent-obs}. Importantly, Subsection~\ref{subsec:latent-state-seq} (Inferring Latent State Sequence) only requires this observation conditional independence and does not require transition independence.
In contrast, Subsection~\ref{subsec:env-state-seq} (Inferring Environment State Sequence) and Section~\ref{sec:env-secret} (Environment Secret Estimation) adopt the stronger transition- and observation independence conditions in Assumption~\ref{assumption:independent}. We highlight these assumptions explicitly to distinguish the settings covered by each result.

\section{Preliminary}
\noindent\textbf{Notations.}
The set of real numbers is denoted by $\reals$. Random variables will be denoted by capital letters, and their realizations by lowercase letters (\eg, $X$ and $x$).  A sequence of random variables and their realizations with length $T$ are denoted as $X_{0:T}$ and $x_{0:T}$. The notation $x_i$ refers to the $i$-th component of a vector $x \in \reals^{n}$ or to the $i$-th element of a sequence $x_0, x_1, \ldots$, which will be clarified by the context. 
Given a finite set $\mathcal{S}$, let $\dist(\mathcal{S})$ be the set of all probability distributions over $\mathcal{S}$. The set $\mathcal{S}^{T}$ denotes the set of sequences with length $T$ composed of elements from $\mathcal{S}$, and $\mathcal{S}^\ast$ denotes the set of all finite sequences generated from $\mathcal{S}$. The empty string in $\mathcal{S}^\ast$ is denoted by $\varnothing$. The notation
$X_1 \perp\!\!\!\perp X_2 |Y$
means random variables $X_1$ and $X_2$ are conditionally independent given random variable $Y$.
 
We model the  interaction between a multi-agent system and its environment using a finite-horizon, decentralized, partially observable Markov decision process (Dec-POMDP) without reward:
\[
\mathcal{M} = \langle  T, \mathcal{N}, \mathcal{S},  A, \mathcal{O}, P, \mu, \{E_i\} \rangle,
\]
where:
\begin{itemize}
\item $T \in \nat$   is the time horizon of the problem;
     \item $ \mathcal{N} = \{1, 2, \ldots, N\}$ is a set of $N$ agents,
    \item $\mathcal{S}$ is the finite set of states. 
     \item $\calA = \prod_{i\in \mathcal{N}}\calA_i$ is the joint action space, where $\calA_{i}$ is the finite action space of agent $i \in \mathcal{N}$. The tuple $a  = \langle a_{1}, a_{2 }, \ldots, a_{N} \rangle$  is called the joint action.
    \item $\calO= \prod_{i\in \mathcal{N}} \calO_i$ is the joint    observation space, where $\calO_{i}$ is the finite observation space of agent $i$. The tuple $o = \langle o_{1}, o_{2}, \ldots, o_{N} \rangle$ of individual observations is called the joint  observation;
    \item $P: \mathcal{S}\times \calA \rightarrow\dist(\calS) $ is the probabilistic transition function;
    \item $\mu$ is the initial state distribution.
    \item $E_i: \calS \rightarrow \dist(\calO_i)$ is the emission/observation function for agent $i$.
\end{itemize}

An admissible solution of a \ac{dec_pomdp} is a decentralized joint policy $\bm{\pi}$, i.e., a tuple $\langle \pi_1, \ldots, \pi_N \rangle$, where the individual policy $\pi_i: (\calO_i\times \calA_i)^\ast \calO_i \rightarrow \dist(\calA_i)$ of each agent $i$ maps individual \textbf{observation-action} sequences $y_{i,0:t} = ( o_{i,0}, a_{i,0},  o_{i,1}, \ldots,    o_{i,t})$ to an individual action distribution. It is noted that the individual agent knows its own action but not others' actions.

We introduce the notion of submodular functions next. Denote $\Omega$ as a ground set of $n$ data points, $\Omega = \{x_1, x_2, x_3, \ldots, x_n\}$, and a set function $f : 2^\Omega \rightarrow \mathbb{R}$. We say that $f$ is \textit{normalized} if $f(\emptyset) = 0$, and $f$ is \textit{subadditive} if
\[
f(U) + f(V) \geq f(U \cup V),
\]
holds for all $U, V \subset \Omega$. Define the first-order partial derivative (equivalently, the gain) of an element $j \notin U$ in the context $U$ as
\[
f(j \mid U) = f(U \cup \{j\}) - f(U).
\]
The function $f$ is \textit{submodular}~\cite{fujishige2005submodular} if for all $U, V \subset \Omega$, it holds that
\[
f(U) + f(V) \geq f(U \cup V) + f(U \cap V).
\]
An equivalent characterization of submodularity is the \textit{diminishing marginal returns} property, namely
\[
f(j \mid U) \geq f(j \mid V) \quad \text{for all } U \subset V, \, j \notin V.
\]

Next, we introduce the definition of entropy and mutual information. Entropy is commonly employed to quantify the uncertainty about a random variable \cite{khouzani2017leakage}.  The conditional entropy of a  random variable $X_2$ given another random variable $X_1$ is defined by
\[
H(X_2|X_1) =   -\sum_{x_1\in \mathcal{X}}\sum_{x_2\in \mathcal{X}} p(x_1,x_2) \log p(x_2|x_1).
\]
Conditional entropy quantifies the uncertainty of \( X_2 \) given \( X_1 \), and lower values indicate that \( X_2 \) is easier to infer from knowing the value of \( X_1 \).
The mutual information between two random variables \( X_1 \) and \( X_2 \) is defined by
\[
I(X_1; X_2) = \sum_{x_1 \in \mathcal{X}} \sum_{x_2 \in \mathcal{X}} p(x_1, x_2) \log \frac{p(x_1, x_2)}{p(x_1)p(x_2)}.
\]
Mutual information quantifies the amount of information that one random variable contains about another. Higher values indicate a stronger statistical dependency between \( X_1 \) and \( X_2 \), while \( I(X_1; X_2) = 0 \) if and only if \( X_1 \) and \( X_2 \) are independent.

\section{Leveraging Submodularity for Multi-Agent Active Perception}
\label{sec:submodularity}

This section examines which classes of inference objectives in a multi-agent active perception planning exhibit submodularity with respect to the set of agents' observations.

\subsection{Inferring Latent State Sequence}
\label{subsec:latent-state-seq}


In this subsection, we study a class of active perception objectives for minimizing the uncertainty in the estimation of the joint state trajectories. In particular, we are interested in finding an (approximately)-optimal subset $\mathcal{K}$ of the agents to carry out the active perception tasks and computing their individual policies, for some $\mathcal{K} \subset \mathcal{N}$. The agents in $\mathcal{N}\setminus \mathcal{K}$ are excluded from the perception task and their observations do not contribute to the perception objective. 
 Importantly, selecting a subset $\mathcal{K}$ does \emph{not} remove the other agents from the system dynamics: all agents in $\mathcal{N}$ still execute their respective policies and evolve under the original joint transition model. The only change is that the inference objective is evaluated using the collective observations of agents in $\mathcal{K}$ (e.g., for $\mathcal{N}=\{1,\ldots,5\}$ and $\mathcal{K}=\{1,2,4\}$, all five agents act and affect the state evolution, while only $Y_1,Y_2,Y_4$ are used for estimating $S_{0:T}$).

\begin{problem}
\label{prob:opt_prob_states}
    Given the \ac{dec_pomdp} $\mathcal{M}$, determine a subset $\mathcal{K}$  of $k$ agents and design the joint policy $\{\pi_i, i  \in \mathcal{K}\}$ that maximizes the mutual information between  the latent state sequence $X\coloneqq S_{0:T}$ and the observations:
    \begin{align}
    \label{eq:opt-traj}
\optmax_{\mathcal{K}\subset \mathcal{N}, |\mathcal{K}|=k, 
\bm{\pi_K}= \{\pi_k, k \in \mathcal{K}\}} I(X; \bm{Y}_\mathcal{K}, M_{\bm{\pi_K}})
    \end{align}
    where $\bm{\pi_K}$ is the joint policy for the selected subset of agents, and $M_{\bm{\pi_K}}(\mathcal{K})$ is the induced stochastic process given the joint policy $\bm{\pi_K}$ for the selected agents  in $\mathcal{K}$ and arbitrary policies for the non-selected agents in $\mathcal{N}\setminus \mathcal{K}$; $\bm{Y}_\mathcal{K}$ is the collective observations of the selected   $k$ agents.
\end{problem}

The following assumption is made to establish the submodularity in the perception objective with respect to agents' observations.
\begin{assumption}
 \label{assumption:independent-obs}
 For any $s\in S$, for any $i,j \in \mathcal{N}$ where $i\ne j$,   agent $i$'s observation  $E_i(s)$ is conditionally independent from agent $j$'s observation $E_j(s)$ given state $s$. \end{assumption}

In the following, we fix  
a joint policy $\bm{\pi}$, consisting of individual agents'   observation-based policies. Let  the induced stochastic process be $M_{\bm{\pi}} = \{S_t, \{A_{i,t}, i\in \mathcal{N} \}, \{O_{i,t}, i\in \mathcal{N} \}, t \in \nat\}$.   We denote by $Y_i \coloneqq O_{i,0:T}$ the observation sequence received by agent $i$.  The choice of this joint policy $\bm{\pi}$ is arbitrary, i.e., it need not be the optimal policy for a given perception task.

\begin{lemma} 
\label{lma:cond-independent}
Given any $Y_{i } =O_{i,0:T}, Y_{j } = O_{j, 0:T}$  representing agents $i$ and $j$'s observation sequences for a finite horizon $T$ and $   X = S_{0:T} $ be the latent state  sequence.  Under Assumption~\ref{assumption:independent-obs},  $Y_{i }$ and $Y_{j }$ are conditional independent given $X$.
\end{lemma}
\begin{proof}  
See appendix\footnote{You can find the appendix in \url{https://github.com/AronYoung414/multi-agent-active-perception-grid-world.}}.
\end{proof}

 \begin{lemma}
\label{lma:independent}
Let $A \subset \mathcal{N}$ and  $Y_A = \{Y_i : i \in A\}$. 
  For any $Y_j $ where $j \notin A$, 
    $H(Y_j|Y_A, X) = H(Y_j|X)$.
\end{lemma}
\begin{proof}
    Due to the conditional independence between $Y_j$ and any $Y_i  \in Y_A$ given $X$.
\end{proof}

\begin{lemma}
\label{lma:submodular}
 Let $A\subset \mathcal{N}$ and  $Y_A = \{Y_i : i \in A\}$. 
Let
\[
g(Y_A) \coloneqq I(X; Y_A),
\]
where   $I(\cdot;\cdot)$ denotes mutual information between the random state sequence $X$ and the collection of observations $Y_A$. 
The  function $g(\cdot)  $ is monotone and submodular.
\end{lemma}
\begin{proof}
See appendix.
\end{proof}


Under the assumption of independent observations but possibly coupled dynamics, we have proven that the mutual information between the latent state sequence $X$ and a subset  $Y_A$ of observations is monotone submodular in the subset $Y_A$. 

\begin{remark}
The above analysis is based on a pre-defined decentralized policy. It can be used for selecting a subset of agents whose observations are most informative for estimating the joint state trajectory. Leveraging the monotone, submodular property, the original GreedyMax algorithm \cite{nemhauser1978analysis} can find such a subset of agents with $(1-1/e)$ performance guarantee on the suboptimality.
\end{remark}

 \subsection{Inferring Environment State Sequence}
\label{subsec:env-state-seq}
Although mutual information is often used to quantify information gain, directly maximizing the mutual information between the joint state trajectory and a subset of perception agents’ observations may not lead to the intended inference objective. Since \( I(X; Y) = H(X) - H(X|Y) \), maximizing mutual information encourages reducing the uncertainty of the trajectory given observations, measured by \( H(X|Y) \). However, it also rewards increasing the prior entropy \( H(X) \). As a result, this objective can favor policies that induce joint state trajectories with a high initial uncertainty, rather than those that genuinely improve inference accuracy.

Therefore, we consider another objective function that is well-suited for environmental monitoring tasks. 
We focus on a case where a team of agents is deployed to actively monitor a dynamic environment state trajectory. In the \ac{dec_pomdp}, 
 each joint state $s = \langle s_e, s_1, \ldots, s_N\rangle $ consists of an environment agent's state $s_e$ and agent $i$'s state $s_i$ for each $i \in \mathcal{N}$ of the multi-agent system.

We modify the objective function in \eqref{eq:opt-traj} for the environment-state trajectory estimation as follows:
\begin{align}
\label{eq:opt-traj-env}
    \optmax_{\mathcal{K}\subset \mathcal{N}, |\mathcal{K}|=k, 
    \bm{\pi_K}= \{\pi_k, k \in \mathcal{K}\}} I(X_e;\bm{Y}_\mathcal{K}, M_{\bm{\pi_K}})
\end{align} 
where $X_e = S_{e,0:T}$ is the environment state trajectory.

Since the environment state sequence is uncontrollable,  $H(X_e)$ is a constant. In this case, given   $I(X; Y ) =  H(X) -  H(X|Y)$,  
maximizing the mutual information is equivalent to minimizing the uncertainty in the environment state trajectory.
That is, the optimization problem in \eqref{eq:opt-traj-env} is equivalent to 
 \begin{align}
    \label{eq:opt-traj-env-entropy}
\optmax_{\mathcal{K}\subset \mathcal{N}, |\mathcal{K}|=k, 
\bm{\pi_K}= \{\pi_k, k \in \mathcal{K}\}}  - H(X_e| \bm{Y}_\mathcal{K}, M_{\bm{\pi_K}})
    \end{align} 

In general, the objective in \eqref{eq:opt-traj-env} is not submodular in the set of agents' observations. This is because of the local observations of two agents may no longer be conditionally independent given the environment-state trajectory. However, we show that for a subset of the problems, we can also derive submodularity in \eqref{eq:opt-traj-env}.

Moreover, under the transition- and observation-independence conditions below, the Dec-POMDP induced by selecting a subset of agents can be constructed in a standard way by restricting the dynamics and observation model to the environment state and the selected agents' local states. In particular, because the joint transition and observation models factor across agents, the evolution of the selected agents depends only on their own local states/actions and the environment process, and the non-selected agents can be marginalized without affecting the resulting reduced model.

\begin{assumption}
\label{assumption:independent}
    The agents and their environment have independent dynamics. That is, there exists transition functions $P_e, \{P_i, i\in \mathcal{N}\}$ such that
    \begin{multline}
        \label{eq:indepent-dyn}    P( (s_e',s'_1,\ldots, s'_N)| (s_e,s_1,\ldots, s_N), (a_1,\ldots, a_N)) \\
    = P_e(s_e' |s_e  ) \prod_{i=1}^N  P_i(s_i'| s_i, a_i).
    \end{multline}
The initial state distribution can also be decomposed as 
\[
\mu_0((s_e,s_1,\ldots, s_N))= \mu _{0,e}(s_e) \cdot \prod_{i=1}^N \mu_{0,i} (s_i) .
\]
The observation of agent $i$ is independent from the other agents' states: That is, if $s_e=s_e'$ and $s_i=s_i'$, then for any $s_j, s_j'$ when $j \ne i$, 
\[
E_i(s_e,s_1,\ldots, s_N )= E_i(s'_e,s'_1,\ldots, s'_N ).
\]
\end{assumption}
In this case, we define a local observation function $\hat{E}_i: S_e\times S_i \rightarrow \dist(\mathcal{O}_i)$ such that $ \hat{E}_i(s_e,s_i)  \coloneqq E_i(s_e,s_1,\ldots, s_N )$.

Due to the decoupled dynamics and observation, we can define individual agent's observation sequence based on its own policy. Let $\pi_i$ be an observation-dependent policy for agent $i$.  The $\pi_i$-induced local observation $Y_i=  O_{i,0:T}$ is defined by
\[
\Pr(O_{i,t}=o) =  \hat{E}_i(o| S_{e,t}, S_{i,t}),
\]
where $S_{e,t}\sim P_e(\cdot | S_{e,t-1})$, $
S_{i,t} \sim P_i( \cdot| S_{i,t-1}, A_{i,t})$, $S_{i,0}\sim \mu_{0,i}(\cdot)$, and \\
$A_{i,t}\sim \pi_i(\cdot | O_{i,0:t}, A_{i, 0:t-1})$, for any $t \in \{1,\ldots, T\}$.

The following result again assumes a fixed joint policy $\bm{\pi}=\langle\pi_1,\ldots, \pi_N \rangle$ and the induced stochastic process of states, actions, and observations.
\begin{lemma}
\label{lma:cond-independent-special}
For any $i,j\in \mathcal{N}$, $i \ne j$. Let   $Y_i, Y_j$ be the local observation sequences for agents $i, j$.
Under Assumption~\ref{assumption:independent}, let $X_e \coloneqq S_{e, {0:T}}$ be the latent state sequence of the environment agent, then  $Y_i$ and $Y_j$ are conditionally independent given $X_e$.
\end{lemma}

\begin{proof}
See appendix.
\end{proof}

\begin{corollary}
   Let $A\subset \mathcal{N}$ and  $Y_A = \{Y_i : i \in A\}$. 
  For any $Y_j $ where $j \notin A$, 
    $H(Y_j|Y_A, X_e) = H(Y_j|X_e).$
    \end{corollary}
\begin{proof}
The proof is similar to that of Lemma~\ref{lma:independent} given Lemma~\ref{lma:cond-independent-special}.
\end{proof}
\begin{lemma}
\label{lma:submodular-special}
Let $A\subset \mathcal{N}$ and  $Y_A = \{Y_i : i \in A\}$. 
Let
$
g_e(Y_A) \coloneqq I(X_e; Y_A),
$  where  $I(X_e; Y_A)$ is the mutual information between the environmental state sequence $X_e$ and the collective observations $Y_A$. The function $g_e(\cdot)$ is monotone and submodular.   
\end{lemma}
\begin{proof}
Similar to the proof of Lemma~\ref{lma:submodular} where we use Lemma~\ref{lma:cond-independent-special}.
\end{proof}

\subsection{Environment Secret Estimation}
\label{sec:env-secret}
In this section, we consider a perception objective besides the aforementioned environment trajectory estimation. As in Subsection~\ref{subsec:env-state-seq}, our analysis for this setting assumes the transition- and observation-independence conditions in Assumption~\ref{assumption:independent}, so that the Dec-POMDP induced by a reduced agent set can be constructed by restricting to the environment state and the selected agents' local states/observations.
Consider a random variable $Z$ defined as a surjective function of $X_e$   \ie, $Z = f(X_e)$. It is noted that $H(X_e)$ is a constant and hence $H(Z)$ is a constant.

The agent selection and optimal perception planning problem for inferring the value of $Z$ given the collective observations is formulated as the following optimization problem:

\begin{align}
\label{eq:opt-secret}
\optmax_{\mathcal{K}\subset \mathcal{N}, |\mathcal{K}|=k, 
\bm{\pi_K}= \{\pi_k, k \in \mathcal{K}\}} - H(Z| \bm{Y}_\mathcal{K}, M_{\bm{\pi_K}})
\end{align}  

First, based on the relation between mutual information and entropy, we have
$
I(Z; Y_A)   =H(Z) - H(Z|Y_A),
$
and thus
\[
I(X_e; Y_A )-  I(Z; Y_A )=  H(X_e)- H(X_e|Y_A)  - H(Z) +  H(Z|Y_A).
\]
Given the fact that $H(X_e|Y_A)  -    H(Z|  Y_A)  = H(X_e|Z, Y_A)$, we have
\begin{equation}
\label{eq:secret_MI}
I(Z;Y_A) = I(X_e; Y_A )-  H(X_e)+ H(Z) + H(X_e|Z, Y_A).
\end{equation}
We will show that $I(Z; Y_A)$ is monotone and submodular in the following lemma. 
 \begin{lemma}[Approximate Submodularity]
\label{lma:approx_submodular_secret}
Let
\[
g(A) \coloneqq I(Z; Y_A),
\]
where $Z$ is a secret variable and $Y_A$ denotes the joint observations collected by agent set $A$.
Then $g(\cdot)$ is monotone and $\epsilon$-approximately submodular, i.e., there exists a submodular function
\[
h(A) \coloneqq I(X_e; Y_A)
\]
such that for all $A \subseteq \mathcal{N}$,
\[
(1-\epsilon) h(A) \le g(A) \le (1+\epsilon) h(A),
\]
where
\[
\epsilon \coloneqq \max_{A} \frac{I(X_e; Y_A \mid Z)}{I(X_e; Y_A)} \in (0,1].
\]
\end{lemma}
 
 \begin{proof}
See appendix.
\end{proof}

\section{Approximation Algorithms for Joint Agent Selection and Policy Synthesis}

The previous section shows the submodularity for several inference objectives under various assumptions for \ac{dec_pomdp}, while the joint policies are assumed to be fixed. Thus, these results do not solve the problems in \eqref{eq:opt-traj}, \eqref{eq:opt-traj-env}, or \eqref{eq:opt-secret}, since they all require the joint agent selection and policy synthesis.

\subsection{Approximation Algorithm  and Performance Guarantee}
In this section, we devise an efficient method called IMAS$^2$ (Algorithm~\ref{alg:greedy_MI}) to solve these problems and analyze the performance guarantee. The variable $X$ to be inferred  can be changed to $X_e$ or $Z$, depending on  \eqref{eq:opt-traj}, \eqref{eq:opt-traj-env}, or \eqref{eq:opt-secret}.

\begin{algorithm}[t]
\caption{IMAS$^2$ Algorithm}
\label{alg:greedy_MI}
\begin{algorithmic}[1]
\REQUIRE Ground set of agents $\mathcal{N}$, budget $k$
\ENSURE Selected agent set $\mathcal{K}$
\STATE Initialize $\mathcal{K}^{(0)} \gets \varnothing$, $\bm{\pi}^{(0)}\gets \varnothing$.
\FOR{$i = 1$ to $k$}
    \STATE Select
    $   ( j^\ast , \pi_{j^\ast} ) = $ \\$ 
   \argmax_{j \in \mathcal{N} \setminus \mathcal{K}^{(i)}}\max_{ \pi_j\in \Pi_j}
        \Big( I(X; \bm{Y}_{\mathcal{K}^{(i)} \cup \{j\}}, M_{\bm{\pi}_{\mathcal{K}^{(i)}} \cup \{\pi_j\}})
        - I(X; \bm{Y}_{\mathcal{K}^{(i)}}, M_{\bm{\pi}_{\mathcal{K}^{(i)}}}) \Big)
  $
    \STATE Update $\mathcal{K}^{(i)}  \gets \mathcal{K}^{(i-1)} \cup \{j^\ast\}$, $\bm{\pi}^{(i)} \gets\bm{\pi}^{(i-1)} \cup \{\pi_{j^\ast}\}  $
\ENDFOR
\STATE \textbf{return} $\mathcal{K} \coloneqq \mathcal{K}^{(k)}, \bm{\pi}^\star \coloneqq \bm{\pi}^{(k)}.$
\end{algorithmic}
\end{algorithm}


 In each iteration of Algorithm~\ref{alg:greedy_MI}, for each candidate agent,  its optimal local policy is computed to maximize the mutual information between the agent’s observations and the latent variables of interest. Then, with the computed policies, the algorithm evaluates the marginal gain contributed by each agent and greedily selects the next one that yields the largest improvement.


A key property of submodular functions is that they admit strong approximation guarantees when optimized under cardinality constraints. However, because we are simultaneously selecting the agent and determining their policy, the proof of approximate optimality needs to be modified.
Next, we derive the performance guarantee of the IMAS$^2$ algorithm when the objective functions are submodular.

 In \eqref{eq:opt-traj},  let  
  $f(A, \bm{\pi}_A ) =  I(X; Y_{A}, \calM_{ \bm{\pi}_A})$ and in \eqref{eq:opt-traj-env}, let $f(A, \bm{\pi}_A ) =  I(X_e; Y_{A}, \calM_{ \bm{\pi}_A})$, 
  for any subset $A\subset \mathcal{A}$ and $\bm{\pi}_A = \{\pi_i, i\in A\}$.
Let $(\calK^\star, \bm{\pi}^\star)$ be the optimal solution to \eqref{eq:opt-traj} (or \eqref{eq:opt-traj-env}). Let $(\calK^{(i)}, \bm{\pi}^{(i)})$ be the solution returned by Algorithm~\ref{alg:greedy_MI} after the $i$-th selection.
By default, $\mathcal{K}^{(0)}  = \varnothing$ and $\bm{\pi}^{(0)} = \varnothing$.  
  By default, $f(\varnothing,\varnothing)=0$.
\begin{proposition}
\label{prop:delta}
	Let $\Delta_i  \coloneqq f(\calK^{(i )}, \bm{\pi}^{(i )})  - f(\calK^{(i-1)}, \bm{\pi}^{(i-1)})   $ for $1\le i < k$. It holds that 
	$ 
	\Delta_{i+1}\le \Delta_{i }, \forall  0 \le i < k.
		$  
\end{proposition}
\begin{proof} 
See appendix.
\end{proof}

\begin{theorem} 
\label{thm:approximate_error}

 Assuming for all $0 \le i <k $, $\frac{\Delta_{i}}{\Delta_{i+1}}   \le \frac{k+1}{k}$, then 
\begin{equation}
\label{eq:greedy-bound}
f(\calK^{(k)}, \bm{\pi}^{(k)}) \ge (1-\frac{1}{e}) f(\calK^\star, \bm{\pi}^\star).
\end{equation}\end{theorem}
\begin{proof}
See appendix.
\end{proof}

\begin{remark}Since the objective function in \eqref{eq:opt-secret} is only approximately submodular, analyzing the performance of the proposed algorithm requires explicitly accounting for the approximation ratio $\epsilon$. We defer a performance analysis and the establishment of formal guarantees for the secret inference case to future work.
\end{remark}
\subsection{Policy Synthesis for Maximizing Information Gain}

In Algorithm~\ref{alg:greedy_MI}, the outer maximization is straightforward: the optimal agent is selected by comparing all candidate agents. However, we must also address the inner policy optimization problem: for each $j \in \mathcal{N}\setminus \mathcal{K}^{(i)}$,
\begin{equation*}
\begin{aligned}
\bm{\pi}_{ j^\ast} & = \arg\max_{\pi_j} \Big( I(X; \bm{Y}_{\mathcal{K} \cup \{j\}}, M_{\bm{\pi}_{\mathcal{K}} \cup \{\pi_j\}}) 
 - I(X; \bm{Y}_{\mathcal{K}}, M_{\bm{\pi}_{\mathcal{K}}}) \Big).
\end{aligned}
\end{equation*} 
We discuss how to solve the optimization problem using existing results in literature.

\paragraph{Late state sequence estimation: } In \cite{Molloy2023smoother}, the authors develop a method to compute an optimal policy that maximizes $I(X; Y, M_\theta)$ for a single-agent POMDP. This method can be used for solving the inner maximization problem: First, we compute the single-agent POMDP for agent $j$ by fixing the selected agents' policies $\bm{\pi}_{\mathcal{K}^{(i)}}$ computed from the previous iteration; second, the algorithm for trajectory estimation in single-agent POMDP \cite{Molloy2023smoother}
can be applied. Note that the observations received by selected agents are included in the observation of the single-agent POMDP.

\paragraph{Inferring environment state trajectory or secret}The method in \cite{shi2024active} employs a policy gradient approach. 
  For agent $i$, Let $\{\pi_{i, \theta} | \theta_i \in \Theta\}$ be a class of parameterized stochastic policies. We denote the joint policy parameter for all agents as $\bm{\theta} $ and the corresponding joint policy as $\bm{\pi}_{\bm{\theta}}$. 
Let $Y(\theta_j)$ be agent $j$'s observation under agent $j$'s policy parameterized by $\theta_i$ and $\bm{Y}(\bm{\theta}_{\mathcal{K}^{(i)}})$ is the joint observation of agents in $\mathcal{K}^{(i)}$ given the joint policy of selected agents. 

 As argued in the section~\ref{subsec:env-state-seq}, maximizing the mutual information for these two cases is equivalent to minimizing the conditional entropy in the unknown variable. 
To obtain the locally optimal policy parameter $\theta_j$, we initialize a policy parameter $\theta_j$ and carry out the gradient descent. At each iteration $\tau\ge 0$, let $\bm{\theta} \coloneqq \bm{\theta}_{\mathcal{K}^{(i)}}$ represent the joint policy of selected agents for notation simplicity,
\begin{equation}
\label{eq:gradient_decsent_algorithm}
\theta_{j}^{ \tau + 1}  = \theta^{\tau}_j - \eta  [\nabla_{\theta_j} H (X_e |\bm{Y}(\bm{\theta}) \cup Y (\theta_j))| _{\theta_j = \theta_j^\tau}],
\end{equation}
where $\eta$ is the step size (learning rate). 
 The gradient of conditional entropy is given by
 \begin{equation}
\label{eq:entropy-grad-expectation}
 \nabla_{\theta_j} H (X_e | \bm{Y}(\bm{\theta})) =  \expect_{\bm{y} \sim \calM_{\bm{\theta}}} \left[ H(X_e|\bm{Y} = \bm{y}) \nabla_{\theta^j} \log P_\theta(y_j)\right].
 \end{equation}
Notably, when   evaluating the conditional entropy, we use the joint observation $\bm{y}$ including both the observations of selected agents and the agent $j$'s observation under the current policy $\theta_j^\tau$.  When we are evaluating the gradient of the log probability of observations, we only use the observations of agent $j$ because the other selected agents' policies are fixed and independent from $\theta_j$.
 We have the following result from \cite{shi2024active}:
\begin{equation}
\label{eq:two_grads_equals}
\begin{aligned}
\nabla_{\theta_j} \log  P_\theta(y_j)=  \sum_{t = 0}^{T} 
\nabla_{\theta_j} \log  \pi_{\theta_j}(a_{t}| o_{j,0:t-1}),
\end{aligned}
\end{equation}
where $o_{j,0:t-1}$ is the local  observation of agent $j$. The policy gradient from minimizing $H(Z|\bm{Y}(\bm{\theta}))$ is derived analogously (see \cite{shi2024active} for technical details.)


\section{Experiment}

\begin{figure}[t]
    \centering
    \includegraphics[width=0.7\linewidth]{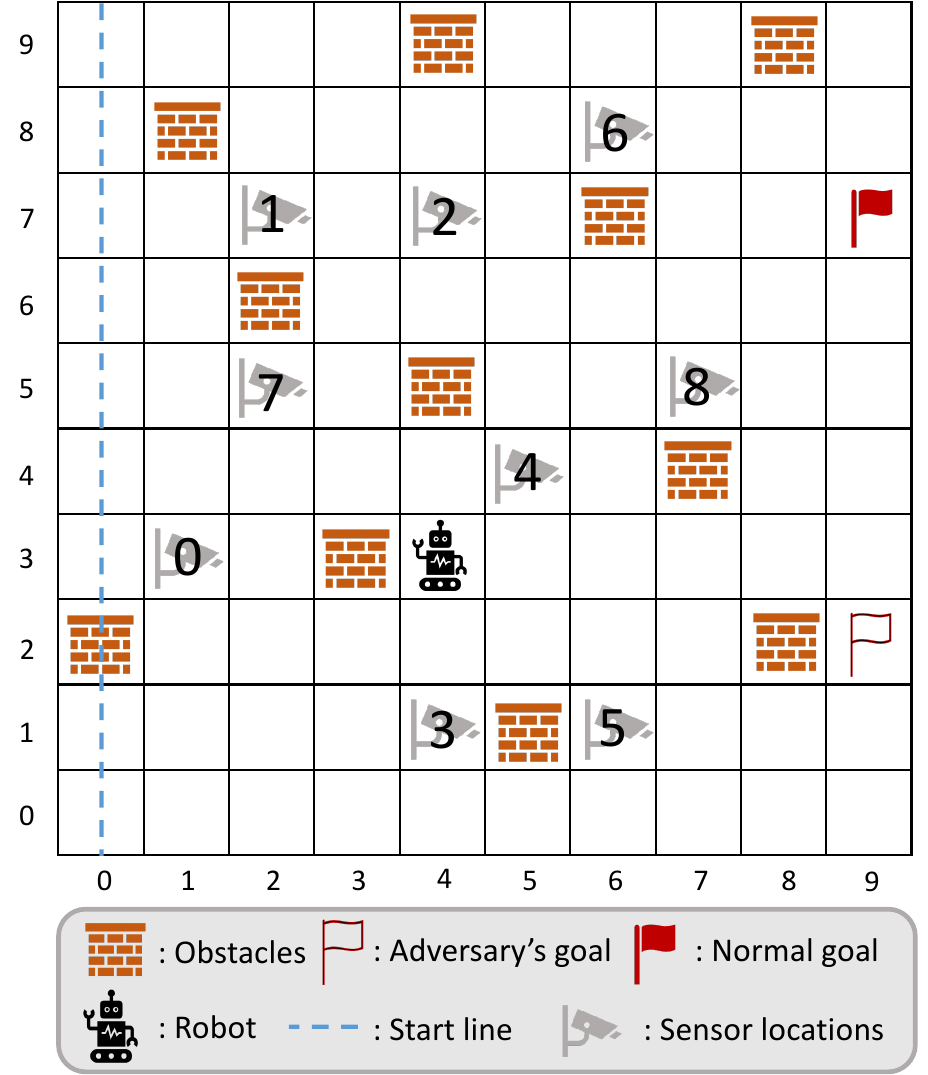}
    \caption{The $10 \times 10$ grid world environment.}
    \label{fig:grid_world}
    \Description{grid world}
\end{figure} 

\begin{figure}[t]
    \centering
\includegraphics[width=0.6\linewidth]{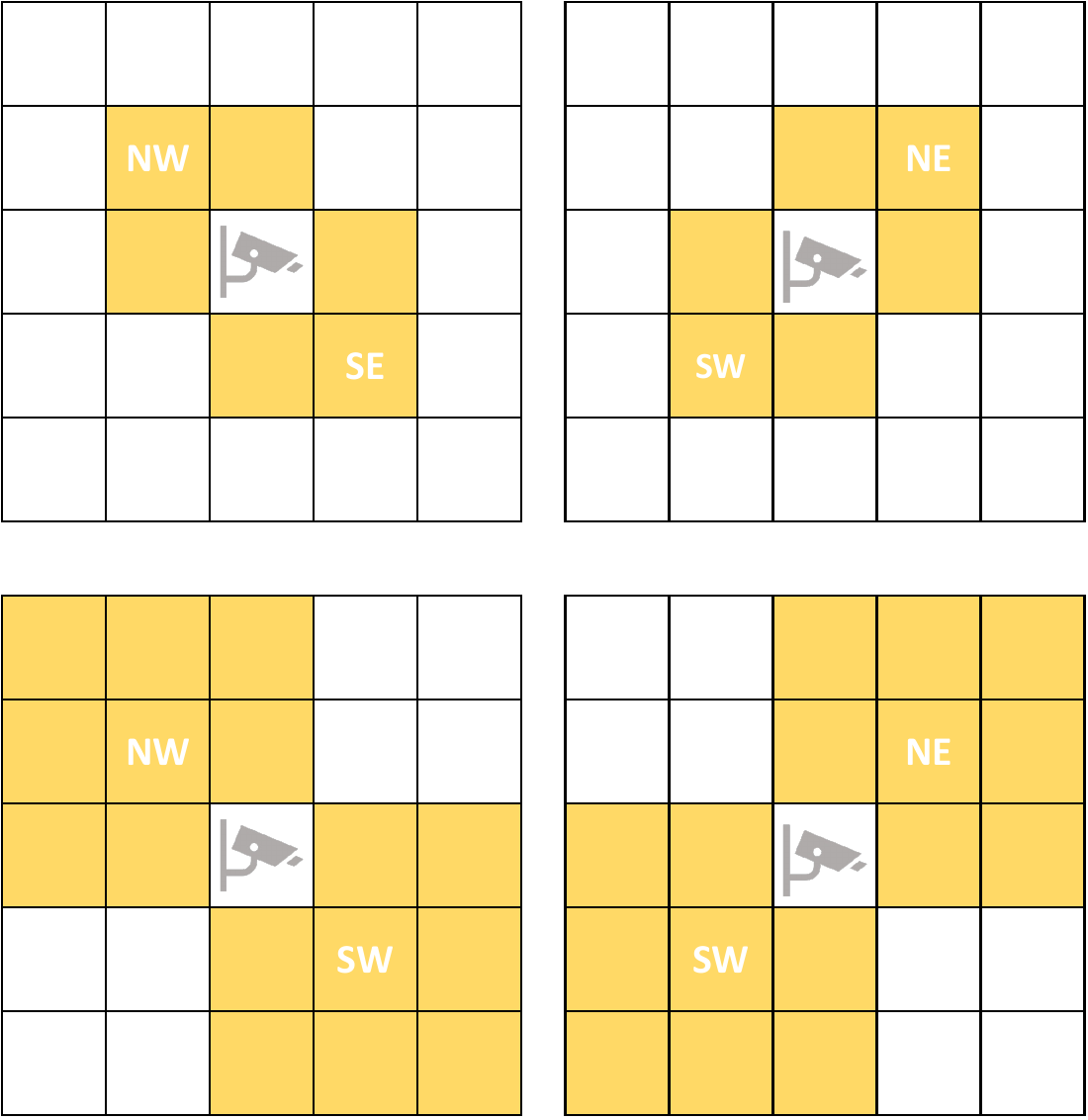}
    \caption{The sensor range of different actions (Top: small range sensors; Bottom: large range sensors).}
    \label{fig:sensor_range}
    \Description{sensor range}
\end{figure}

\begin{figure}[t]
    \centering
    \begin{subfigure}[t]{0.21\textwidth}
        \centering
        \includegraphics[width=\linewidth]{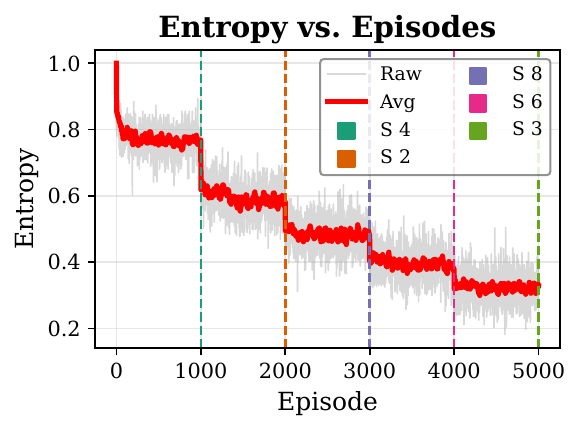}
        \caption{Convergence results for the sensing policies of the selected five sensors.}
        \label{fig:results_9_sensors}
    \end{subfigure}
    \hfill
    \begin{subfigure}[t]{0.21\textwidth}
        \centering
        \includegraphics[width=\linewidth]{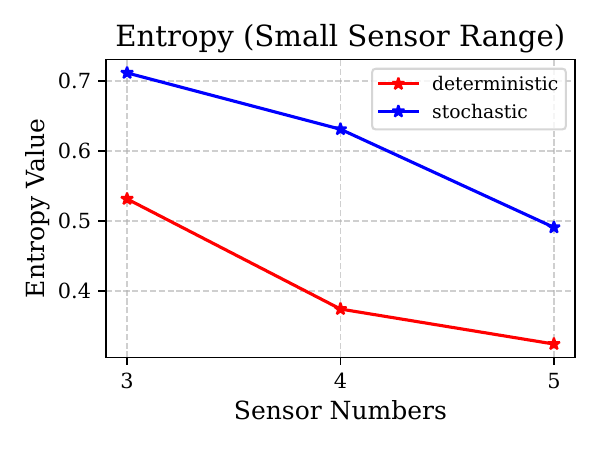}
        \caption{The final entropy under a small sensor range setting.}
        \label{fig:small_sensor}
    \end{subfigure}
    \hfill
    \begin{subfigure}[t]{0.21\textwidth}
        \centering
        \includegraphics[width=\linewidth]{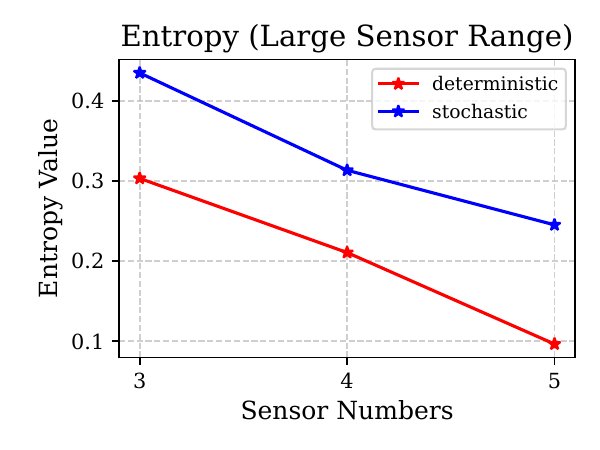}
        \caption{The final entropy under a large sensor range setting.}
        \label{fig:large_sensor}
    \end{subfigure}
    \hfill
    \begin{subfigure}[t]{0.21\textwidth}
        \centering
        \includegraphics[width=\linewidth]{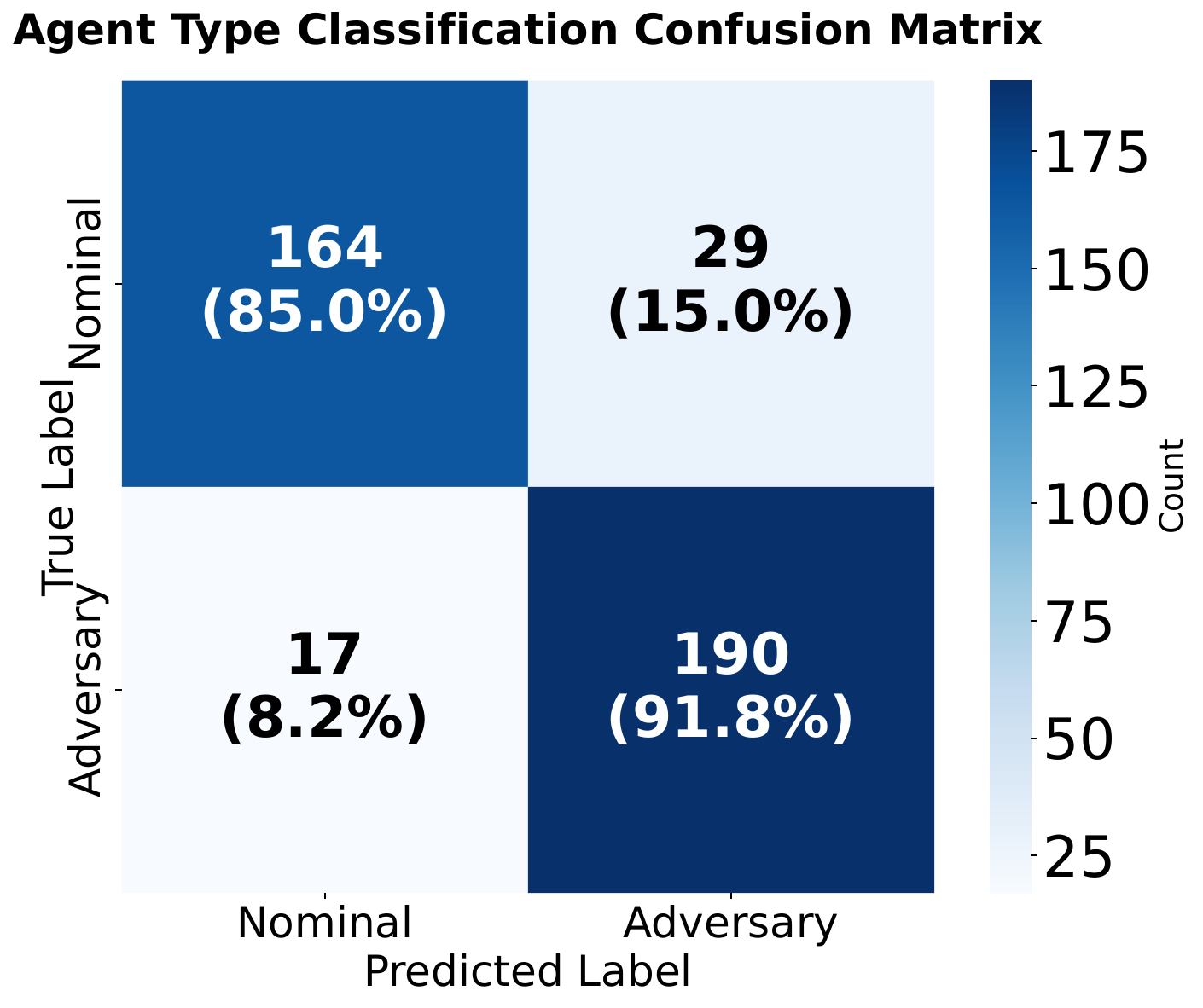}
        \caption{The confusion matrices of the type inference under a deterministic environment.}
        \label{fig:confusion_deterministic}
    \end{subfigure}
    \caption{Comparison of convergence, entropy, and inference performance under different sensing conditions.}
    \label{fig:four_results}
      \Description{Comparison of convergence, entropy, and inference performance under different sensing conditions.}
\end{figure}

We illustrate the effectiveness of our IMAS$^2$ algorithm using a multi-agent selection and sensing strategy design problem in a $10 \times 10$ grid-world environment (Figure~\ref{fig:grid_world}). 
The multi-agent system is a directed sensor network. 
The underlying state of the environment includes a tuple $s_e= (x,y, \mathbf{b})$ where $x, y$ are the position of an environment agent, marked by the robot, which are not controlled. The last variable $\mathbf{b}$ is a Boolean variable. If $\mathbf{b}=0$, then the robot is benign and follows a stochastic shortest path policy to reach the ``Normal goal'' (marked by red, filled flag). If $\mathbf{b}=1$, then the robot is adversarial and follows a stochastic shortest path policy to reach the ``Adversary's goal'' (marked by red, unfilled flag).  Once a robot reaches its goal, it remains there. 
The boundary and obstacles are bouncing obstacles. The robot can move north, south, east, or west, or choose to stay in place. The initial position of the robot is uniformly distributed on the first column of the grid world, and the type of the robot is uniformly sampled from a Bernoulli distribution for $\mathbf{b}$ with $p=0.5$.

In the later experiments, we consider two types of robot dynamics: deterministic and stochastic. Under deterministic dynamics, an action directly results in the intended movement (e.g., moving north always results in a northward transition). Under stochastic dynamics, however, when an agent attempts to move in a given direction, there is a $0.1$ probability of deviating to one of the orthogonal directions. For example, an intended eastward movement may instead result in a northward or southward move, each with probability $0.1$. Collisions with obstacles or boundaries cause the agent to remain in its current cell.

\paragraph{Partial observable multi-agent system} The underlying environment state $s_e$ is partially observed by a multi-agent system, which is a directed sensor network. The sensor can change its limited field of view by choosing different angles of viewing. 
The sensing ranges of these sensors are illustrated in the top two panels of Figure~\ref{fig:sensor_range}. Each sensor has four possible sensing actions: NW, NE, SW, and SE. A sensor may detect the robot if the robot lies within the chosen sensing direction. If the robot is detected, the sensor will receive the exact robot location. We consider a weak sensor scenario. The detection process is stochastic: if the robot is within range, it is detected with probability $0.5$; otherwise, with probability $0.5$, the sensor receives a null observation ($n$). When the robot is outside the sensing range, the sensor also receives a null observation with probability $1$. Note that this multi-agent system, including the uncontrollable environment agent,  can be modeled as a \ac{dec_pomdp} satisfying Assumption~\ref{assumption:independent}.

\paragraph{Inference objective} 
We denote $Z$ be the estimate of the unknown robot's type. This corresponds to the case with inferring the environment secret (section~\ref{sec:env-secret}) $Z = f(S_{e,0:T})$ is defined such that $f(S_{e,0:T}) = \mathbf{b}$ if $S_{e,0} = (x,y, \mathbf{b})$ for any initial state $(x,y)$. The objective is twofold: 1. select $5$ locations for allocating sensors, and 2. compute an optimal sensing policy for each selected sensor such that the collective observations of the sensors for inferring the type of the robot.

\paragraph{Observation-based policy space} For each selected sensor, we employ a long short-term memory (LSTM) neural network-based policy that directly processes observation sequences to action decisions. We implement Algorithm~\ref{alg:greedy_MI} in which the individual sensor's optimal sensing strategy is computed using the policy gradient method in \cite{shi2024active}. 
Figure~\ref{fig:results_9_sensors} illustrates the convergence trend in a setting where the dynamics are deterministic and the sensor range is small
\footnote{We sample $M = 100$ trajectories and set the horizon $T = 15$ for each iteration. The fixed learning rate of the policy gradient algorithm is set to be $0.001$. The hidden dimensions of all layers are set to be $64$. We run $N = 1000$ iterations for each sensor on the 12th Gen Intel(R) Core(TM) i7-12700; the average time consumed for one iteration is $1.3$ seconds. The code is available at \url{https://github.com/AronYoung414/multi-agent-active-perception-grid-world.}}.

Because maximizing mutual information is equivalent to minimizing conditional entropy, which is often more intuitive to interpret, we present the experimental results in terms of conditional entropy. Initially, when no sensors are allocated, the entropy $H(Z)$ equals 1, corresponding to the uncertainty in the initial type distribution. After the first iteration, sensor 4 is selected, and its optimal perception policy reduces the conditional entropy (given its observation) to $0.78$. In the second iteration, sensor 2 is selected, and the joint observations from sensors 4 and 2 further reduce the conditional entropy to $0.6$.
We observe that the policy parameter for each selected sensor converges quickly, within the first 100 iterations. This explains the drops shown in Fig.~\ref{fig:results_9_sensors}.

This iterative process continues until all five sensors are allocated. Upon convergence, the conditional entropy $H(Z | Y; \theta^\star)$, where $\theta^\star$ denotes the parameters of the computed optimal decentralized policies, is approximately $0.367$. This result indicates that the collective observations provide substantial information about the ground robot’s type on average. The selected sensors are $2, 3, 4, 6, 8$—positions that are close to the goals. Notably, sensors located near the robot’s initial position were not selected. We also construct confusion matrices to visualize the classification performance (Figure~\ref{fig:confusion_deterministic}). 
These results show that our method enables the sensors to accurately infer the robot’s type in both nominal and adversarial settings. The inference accuracy is high as $88\%$.

We further evaluate our algorithm under different sensing configurations and environment dynamics. Specifically, we consider two   setups: selecting three sensors from five candidate locations (sensor $0, 1, 2, 3, 4$) and selecting four sensors from seven candidate locations (sensor $0, 1, 2, 3, 4, 5, 6$). Both deterministic and stochastic motion models are tested. In addition, we vary the sensing range—small versus large—as illustrated in the bottom two panels of Figure \ref{fig:sensor_range}.

Figures \ref{fig:small_sensor} and \ref{fig:large_sensor} summarize the resulting conditional entropies under these configurations. For the small sensing range (Figure \ref{fig:small_sensor}), the entropy decreases from about $0.53$ to $0.32$ in the deterministic setting and from $0.70$ to $0.48$ in the stochastic setting as the number of sensors increases from $3$ to $5$. This indicates that adding sensors consistently improves the ability to infer the robot’s type or trajectory. However, the entropy values remain relatively high because of the limited coverage and observation overlap.

When the sensor range is enlarged (Figure \ref{fig:large_sensor}), the performance improves significantly. The entropy drops from roughly $0.32$ to $0.09$ in the deterministic case and from $0.43$ to $0.24$ in the stochastic case as the number of sensors increases. These results highlight two important trends: 1. deterministic environments yield lower residual uncertainty than stochastic ones, since robot behavior is more predictable; and 2. wider sensing coverage substantially enhances information gain, even with the same number of deployed sensors. Overall, the quantitative results demonstrate that the proposed IMAS$^2$ algorithm effectively balances agent selection and policy optimization to reduce uncertainty in cooperative perception tasks.

\begin{table}[t]
\centering
\caption{Baseline comparison between IMAS$^2$ and IPG.}
\label{table:baseline}
\begin{tabular}{@{}cccc@{}} 
\toprule
 & \begin{tabular}[c]{@{}c@{}} Resulting \\ Entropy \end{tabular} & \begin{tabular}[c]{@{}c@{}}Inference \\ accuracy\end{tabular} & \begin{tabular}[c]{@{}c@{}}Time per \\ iteration\end{tabular} \\ \midrule
IMAS$^2$ & 0.493 & 86.0\% & 1.58 s \\
Fixed Selector (IPG) & 0.525 & 75.5\% & 7.62 s \\
Random Selector (IPG) & 0.558 & 70.7\% & 7.63 s \\ 
Visibility-Based Selector (IPG) & 0.502 & 84.1\% & 7.63 s \\ \bottomrule
\end{tabular}
\end{table}

\textbf{Baseline Comparison.}  Existing approaches cannot directly address the joint problem of sensing-agent selection and decentralized active perception. The sensing approaches in~\cite{albrecht2017reasoning, shafipour2022online}
employ Bayesian inference mechanisms similar to those used in IMAS$^2$. However, these works focus on \emph{passive} observation models and do not design or optimize \emph{active sensing policies} that control how information is gathered. And the key challenge is that our objective is \emph{mutual information}, rather than a cumulative reward/value function. Consequently, standard value-function-based MARL algorithms (e.g., MADDPG~\cite{lowe2017multi}, MAPPO~\cite{yu2021mappo}) are not directly applicable, and classical \ac{dec_pomdp} solvers require an explicit reward structure or a belief-based value function.
We therefore compare against a variant of the Independent Policy Gradient (IPG) method~\cite{Daskalakis2020ipg}, which is a decentralized gradient-based method that optimizes each agent's policy directly via policy gradients without relying on a centralized value function.
Since IPG does not include a mechanism for selecting an optimal sensor subset, we provide it with a fixed  chosen group of sensors (\(1, 3, 5, 6, 7\)), randomly chosen group (\(8, 5, 4, 6, 1\)) of sensors, and visibility-based selector (The set (\(0, 2, 4, 5, 8\)) covers the most area) from the nine available candidates. Both methods are evaluated under the stochastic environment with a large sensor range.

 As shown in Table~\ref{table:baseline}, the proposed IMAS$^2$ algorithm achieves a lower conditional entropy (\(0.493\)) compared to IPG (\(0.525\), \(0.558\), \(0.502\)), indicating improved estimation accuracy. Actually, IMAS$^2$ gives a higher inference accuracy ($86.0\%$) in the test environment compared to IPG ($75.5\%$, $70.7\%$, $84.1\%$). Moreover, IMAS$^2$ converges substantially faster, requiring only \(1.5\,\text{s}\) per iteration—approximately \(5.06\) times faster than the IPG baseline (\(7.6\,\text{s}\)). These results demonstrate that IMAS$^2$ effectively balances computational efficiency and information gain in decentralized active perception tasks.

\section{Conclusion}
This paper presented a unified framework for joint agent selection and decentralized policy synthesis in cooperative active perception under the \ac{dec_pomdp} setting. By formulating the perception objective in terms of mutual information and conditional entropy, we established that, under mild independence assumptions, the resulting objective is monotone and submodular with respect to the subset of selected agents' policies. Leveraging this property, we developed the IMAS$^2$ algorithm, which combines submodular optimization with algorithms for active perception planning. Theoretical analysis showed that under a condition for subsequent maximal marginal gains, the proposed algorithm ensures a tight $(1 - 1/e)$ approximation guarantee despite the infinite continuous policy space.
Our experiments in stochastic and deterministic grid-world environments validated the approach, demonstrating that the method effectively solves agent selection and policy optimization to minimize uncertainty in cooperative perception tasks.

Future research could explore the extension to continuous-state and continuous-action  Dec-POMDPs and practical applications such as environment monitoring, intrusion detection, or target tracking. Another direction is to investigate scenarios where the perception agents possess imprecise or uncertain knowledge of the model dynamics, requiring robust or adaptive extensions of the framework. Finally, extending the approach to continuous observation spaces—such as camera images or rich sensory data—would further broaden its applicability to real-world multi-robot and autonomous perception systems.





\begin{acks}
Research was sponsored by the Army Research Laboratory under
Cooperative Agreement Number W911NF-25-2-0045 and by the Army Research Office under Award Number W911NF-22-1-0166.  ARO, as the Federal awarding agency, reserves a royalty-free, nonexclusive
and irrevocable right to reproduce, publish, or otherwise use this software for Federal purposes, and to authorize
others to do so in accordance with 2 CFR 200.315(b). The views and conclusions 
contained in this document are those of the authors and should not be interpreted as representing the official 
policies, either expressed or implied, of the Army Research Laboratory or the U.S. Government. 
\end{acks}

\clearpage

\bibliographystyle{ACM-Reference-Format} 
\bibliography{sample}


\newpage

\appendix

\section{Proofs of Propositions, Lemmas, and Theorems}

\noindent\textbf{Proof of Lemma~\ref{lma:cond-independent}} The sequence $Y_i$ consists of an interleaving sequence of observations $O_{i,0:T}$ and a sequence of actions $A_{i,0:T-1}$. 
 Because $O_{i,t} = E_i(S_t=s)$ is independent from $O_{j,t} = E_j(S_t=s)$ for any $s\in S$  and any $t\ge 0$ by Assumption~\ref{assumption:independent-obs},  the observation $O_{i,0:T}, O_{j, 0:T}$ are conditionally independent given the state  trajectory $X$. Furthermore, since $A_{i,t}$ only depends $O_{i,0:t-1}$ and $A_{i,0:t-1}$, it is independent from $O_{j,0:t-1}$. As a result,  $Y_i, Y_j$ are conditionally independent given $X$.

\noindent\textbf{Proof of Lemma~\ref{lma:submodular}}
Let  $\Omega \coloneqq \{Y_i, i \in \mathcal{N}\}$ be the set of observations for all agents. $Y_A\subset Y_B\subset \Omega$, and $Y_j \in \Omega\setminus Y_B$.

To show $g(\cdot)$ is monotone, we need to show $I(X; Y_A) \le I(X; Y_B)$.
It is noted that 
\begin{align*}
     & I(X; Y_B)- I(X;Y_A) 
    =   I(X; Y_B\setminus Y_A| Y_A)
\end{align*}
where $I(X; Y_B\setminus Y_A| Y_A)$ is the conditional mutual information between $X$ and $Y_B\setminus Y_A$. 

Since conditional mutual information is always non-negative~\cite{Cover2006EIT}, $I(X; Y_B)- I(X;Y_A) \ge 0$ and thus the function $g(\cdot )$ is monotone.
To show $g(\cdot)$ is submodular, we need to show that
\[ I(X; Y_A \cup \{Y_j\}) - I(X; Y_A ) \ge I(X; Y_B \cup \{Y_j\}) - I(X; Y_B).\]
And based on the chain rule of mutual information, 
\[I(X; Y_A \cup \{Y_j\}) - I(X; Y_A )  = I(X; Y_j  |Y_A).\]
Using the property that conditional entropy is monotone, we can derive the following sequence of inequalities:
    \begin{align*}
&    H(Y_j| Y_A) \ge H(Y_j|Y_B)\\
& H(Y_j| Y_A)  - H( Y_j|X) \ge H(Y_j|Y_B) - H( Y_j|X)\\
& H(Y_j| Y_A)  - H( Y_j|X , Y_A) \ge H(Y_j|Y_B) - H( Y_j|X, Y_B)\\
& I(Y_j; X| Y_A) \ge I(Y_j; X|Y_B)\\
& I(X; Y_j|Y_A) \ge I(X; Y_j |Y_B),
    \end{align*}
where the third step is due to Lemma~\ref{lma:independent} and the last step is because conditional mutual information is symmetric.

\noindent\textbf{Proof of Lemma~\ref{lma:cond-independent-special}}
 We proceed by induction on $t$.

\textbf{Base case} $(t=0)$. By Assumption~\ref{assumption:independent} the initial agent states $S_{i,0}$ and $S_{j,0}$ are independent. Since $O_{i,0}$ depends only on $(S_{i,0},S_{e,0})$ and $O_{j,0}$ only on $(S_{j,0},S_{e,0})$, we have
\[
O_{i,0} \perp\!\!\!\perp O_{j,0} \;\big|\; S_{e,0}.
\]
Actions $A_{i,0}$ and $A_{j,0}$ are generated from local policies that depend only on the local observations $O_{i,0}$ and $O_{j,0}$, respectively; therefore
\[
A_{i,0} \perp\!\!\!\perp A_{j,0} \;\big|\; S_{e,0}.
\]
Combining these gives
\[
Y_{i,0}\triangleq (O_{i,0},A_{i,0}) \perp\!\!\!\perp Y_{j,0}\triangleq  (O_{j,0},A_{j,0}) \;\big|\; S_{e,0}.
\]

\textbf{Induction step.} Assume for some $t\ge0$ that
\[
Y_{i,0:t} \perp\!\!\!\perp Y_{j,0:t} \;\big|\; S_{e,0:t}.
\]
We show
\[
Y_{i,0:t+1} \perp\!\!\!\perp Y_{j,0:t+1} \;\big|\; S_{e,0:t+1}.
\]

By the induction hypothesis and the fact that policies are local (i.e., $A_{i,t }$ depends only on $Y_{i,0:t}$), we have
\[
A_{i,t } \perp\!\!\!\perp A_{j,t } \;\big|\; S_{e,0:t}.
\]
Using the factorization in Assumption~\ref{assumption:independent}, conditional on $S_{e,t}$ the next local states are independent across agents:
\[
S_{i,t+1}\ \perp\!\!\!\perp\ S_{j,t+1}\ \big|\ S_{e,t},\;A_{i,t},A_{j,t}.
\] 

Finally, local observation $O_{i,t+1}$ depends only on $(S_{i,t+1},S_{e,t+1})$, and similarly for $O_{j,t+1}$. Thus, conditioned on $S_{e,0:t+1}$, the observations $O_{i,t+1}  $ and $O_{j,t+1} $ are independent given the induction assumptions and the factorization property. Concatenating past and new observations/actions yields
\[
Y_{i,0:t+1} \perp\!\!\!\perp Y_{j,0:t+1} \;\big|\; S_{e,0:t+1},
\]
completing the induction.

\noindent\textbf{Proof of Lemma~\ref{lma:approx_submodular_secret}}
From~\eqref{eq:secret_MI}, the mutual information between the secret $Z$ and the observations $Y_A$
can be written as
\begin{equation}
\label{eq:secret-mi-decomp}
I(Z; Y_A)
= I(X_e; Y_A) - H(X_e) + H(Z) + H(X_e \mid Z, Y_A).
\end{equation}
Rearranging~\eqref{eq:secret-mi-decomp}, we obtain
\[
I(X_e; Y_A) - I(Z; Y_A)
= H(X_e) - H(Z) - H(X_e \mid Z, Y_A).
\]
Since $Z$ is a deterministic function of $X_e$, we have $H(Z \mid X_e) = 0$.
By expanding the conditional mutual information,
\[
I(X_e; Y_A \mid Z)
= H(X_e \mid Z) - H(X_e \mid Z, Y_A)
= H(X_e) - H(Z) - H(X_e \mid Z, Y_A).
\]
Therefore,
\[
I(Z; Y_A) = I(X_e; Y_A) - I(X_e; Y_A \mid Z).
\]
Define
\[
\epsilon \coloneqq \max_{A} \frac{I(X_e; Y_A \mid Z)}{I(X_e; Y_A)}.
\]
Then for all $A$,
\[
I(Z; Y_A) \ge (1-\epsilon) I(X_e; Y_A).
\]
Moreover, by the data processing inequality and the fact that $Z$ is a function of $X_e$,
\[
I(Z; Y_A) \le I(X_e; Y_A).
\]
Combining the above inequalities yields
\[
(1-\epsilon) I(X_e; Y_A) \le I(Z; Y_A) \le I(X_e; Y_A),
\]
which can be loosened to
\[
(1-\epsilon) I(X_e; Y_A) \le I(Z; Y_A) \le (1+\epsilon) I(X_e; Y_A).
\]
Since $I(X_e; Y_A)$ is monotone and submodular in $A$, this shows that $I(Z; Y_A)$ is
$\epsilon$-approximately submodular.

\noindent\textbf{Proof of Proposition~\ref{prop:delta}} 
Recall $(j^\ast, \pi_{j^\ast})$   
  is the choice of agent and its policy selected at the $i$-th iteration of Algorithm~\ref{alg:greedy_MI}. 
According to Algorithm~\ref{alg:greedy_MI}, 
	\begin{align*}
		\Delta_{i+1} & = f(\calK^{(i)}, \bm{\pi}^{(i)})  - f(\calK^{(i-1)}, \bm{\pi}^{(i-1)})  \\
		=  &   f(\calK^{(i-1)}\cup \{j^\ast \}, \bm{\pi}^{(i-1)} \cup \{\pi_{j^\ast} \})  - f(\calK^{(i-1)}, \bm{\pi}^{(i-1)}) \\
		\overset{(i)}{\le} &   f(\calK^{(i-2)}\cup \{j^\ast \}, \bm{\pi}^{(i-2)} \cup \{\pi_{j^\ast } \})  - f(\calK^{(i-2)}, \bm{\pi}^{(i-2)})  \\
	\overset{(ii)}{	\le}  &  \max_ {j \in \mathcal{N}\setminus \calK^{(i-2)}}\max_{ \pi_j\in \Pi_j} \big(f(\calK^{(i-2)}\cup \{j\}, \bm{\pi}^{(i-2)} \cup \{\pi_j\})  \\
		&  - f(\calK^{(i-2)}, \bm{\pi}^{(i-2)}) \big) \\
        = & f(\calK^{(i-1)}, \bm{\pi}^{(i-1)})- f(\calK^{(i-2)}, \bm{\pi}^{(i-2)}) 
		=    
		\Delta_{i}.
	\end{align*}
Here, $(i)$ follows from the \emph{submodularity} of the objective function:  
the marginal contribution of agent $j^\ast$ using policy $\pi_{j^\ast}$ with respect to 
$(\mathcal{K}^{(i-2)}, \bm{\pi}^{(i-2)})$ is greater than or equal to its contribution 
with respect to $(\mathcal{K}^{(i-1)}, \bm{\pi}^{(i-1)})$ under the same policy.  And $(ii)$ is due to the greedy choice made at step $i-1$.

\noindent\textbf{Proof of Theorem~\ref{thm:approximate_error}} 
 For each $i \ge 1$, let the gap between the value under the optimal solution and the value under the solution returned by the $i$-th selection be $g_i \coloneqq f(\calK^\star, \bm{\pi}^\star) -   f(\calK^{(i)}, \bm{\pi}^{(i)})$.

We first show $g_i \le k  \Delta_{i+1}$,  for all $0\le i < k$, using induction: 
At step $0$, 
\begin{align*} 
	g_0 & = f(\calK^\star, \bm{\pi}^\star) -  f(\varnothing, \varnothing)\\
	 \le &\sum_{j \in \calK^\star  } \left(f(j,    \pi_j^\star) -   f(\varnothing, \varnothing) \right) \\
	 \le  & k  \max_{j \in \mathcal{N}} \max_{\pi_j\in \Pi_j}  \left(f(j,    \pi_j ) -   f(\varnothing, \varnothing)\right) = k \Delta_1.
	\end{align*}
where the first inequality is because the contribution of the set of agents and their policies is less than or equal to the sum of individual contributions (due to submodularity).

Assume $g_i \le k \Delta_{i+1}$, that is 
$
f(\calK^\star, \bm{\pi}^\star) -   f(\calK^{(i)}, \bm{\pi}^{(i)})  
\le k \Delta_{i+1}.$
We derive that 
\begin{align*}
	  g_{i+1} \coloneqq & f(\calK^\star, \bm{\pi}^\star) -   f(\calK^{(i+1)}, \bm{\pi}^{(i+1)}) \\
\overset{(i)}{\le} & k  \Delta_{i+1}  + 
f(\calK^{(i)}, \bm{\pi}^{(i)}) - f(\calK^{(i+1)}, \bm{\pi}^{(i+1)}) \\
\overset{(ii)}{=} & k \Delta_{i+1} - \Delta_{i+2}\\
= &   \; k \Delta_{i+2} - (k+1)\Delta_{i+2} + k \Delta_{i+1}\\
\overset{(iii)}{\le}&    k \Delta_{i+2}.
\end{align*}  
where $(i)$ is based on the induction hypothesis, $(ii)$ uses the definition of $\Delta_2$, and $(iii)$
 is because of the assumption that for any $0 \le i <k $, $\frac{\Delta_{i}}{\Delta_{i+1}}   \le \frac{k+1}{k}$. The induction step is proven.

  Using successive gaps, $  \Delta_{i+1} \ge \frac{1}{k} g_i$ and $g_{i+1} =   f(\calK^\star, \bm{\pi}^\star) -   f(\calK^{(i+1)}, \bm{\pi}^{(i+1)})  =f(\calK^\star, \bm{\pi}^\star) -f(\calK^{(i+1)}  , \bm{\pi}^{(i+1)}) +  f(\calK^{(i)}, \bm{\pi}^{(i)}) -   f(\calK^{(i)}, \bm{\pi}^{(i)})   = g_i - \Delta_{i+1} \le g_i - \frac{1}{k} g_i=(1-\frac{1}{k}) g_i$. (This part of the proof is similar to the original proof in \cite{nemhauser1978analysis}).

Consequentially,  $g_k \le  (1-\frac{1}{k})^k g_0 $ where $g_0 = f(\calK^\star, \bm{\pi}^\star)$. Using the exponential bound, $(1-\frac{1}{k})^k\le e^{-1}$ for any $k\ge 1$. Equation~\eqref{eq:greedy-bound} is then derived.

\end{document}